\documentclass[letterpaper, 10pt, conference]{ieeeconf} 
\pdfoutput=1
\IEEEoverridecommandlockouts
\usepackage{graphicx}
\usepackage{url}
\usepackage{xcolor}
\usepackage{amsmath}
\usepackage{mathtools}
\usepackage{tikz}
\usepackage{verbatim}
\usepackage{subcaption}
\usepackage{pgfplots}
\pgfplotsset{compat=1.10}
\usepackage{graphicx}
\usepackage{epstopdf}
\usepackage{amssymb}
\usepackage{relsize}
\usepackage[textwidth=3.8em,textsize=scriptsize,disable]{todonotes}
\usepackage{algorithm}
\definecolor{Gray}{gray}{0.90}
\usepackage{colortbl}
\DeclareMathOperator*{\argmax}{arg\,max}

 \usepackage{url}
 \usepackage{hyperref}
\usepackage{algorithm}
\usepackage{algorithmicx}
\usepackage{algcompatible}
\usepackage{algpseudocode}
\usepackage{color}
\usepackage{blkarray}
\usepackage{cite}

\newtheorem{theorem}{\bf{Theorem}}[section]

\newtheorem{lem}[theorem]{\bf{Lemma}}

\newtheorem{definition}{\bf{Definition}}

\usepackage{amssymb}
\usepackage{url}

\begin{document}
\title{Resilient Strong Structural Controllability in Networks\\ using Leaky Forcing in Graphs}
\author{Waseem Abbas
\thanks{W. Abbas is with the Systems Engineering Department at the University of Texas at Dallas, Richardson, TX (Email: \texttt{waseem.abbas@utdallas.edu}).}
}

\maketitle
\begin{abstract}
This paper studies the problem of selecting input nodes (leaders) to make networks strong structurally controllable despite misbehaving nodes and edges. We utilize a graph-based characterization of network strong structural controllability (SSC) in terms of zero forcing in graphs, which is a dynamic coloring of nodes. We consider three types of misbehaving nodes and edges that disrupt the zero forcing process in graphs, thus, deteriorating the network SSC. Then, we examine a leader selection guaranteeing network SSC by ensuring the accuracy of the zero forcing process, despite $k$ misbehaving nodes/edges. Our main result shows that a network is resilient to $k$ misbehaving nodes/edges under one threat model if and only if it is resilient to the same number of failures under the other threat models. Thus, resilience against one threat model implies resilience against the other models. We then discuss the computational aspects of leader selection for resilient SSC and present a numerical evaluation.
\end{abstract}








\section{Introduction}
\label{sec:intro}
The ability to control a network of agents through external control signals profoundly depends on the underlying network topology~\cite{pasqualetti2014controllability,clark2017submodularity,summers2015submodularity,li2020structural,ruths2014control}. The control signals are injected into the network via a subset of agents, typically called leaders. Computing a minimum number of leaders and their locations within the network for complete controllability are central issues in the cooperative control of networks. A popular approach for the leader selection problem, which is well-studied, is to characterize network controllability from a graph-theoretic perspective. Then formulate an appropriate vertex selection problem in graphs that can be solved through various combinatorial methods. 

To this end, several graph-based characterizations of network controllability have been presented, including those based on graph distances, zero forcing in graphs, dominating sets, matching in graphs, and others~(e.g.,~\cite{monshizadeh2014zero,chapman2013strong,zhang2013upper,jia2019strong,yaziciouglu2016graph,van2017distance,jia2020unifying}). The notion of zero forcing in graphs is particularly useful for strong structural controllability (SSC) in networks, which ascertains that for a given set of leaders, if the network is controllable for one set of edge weights, then it is controllable for all feasible edge weights. The zero forcing phenomenon, introduced in~\cite{work2008zero}, is a dynamic coloring of nodes in a graph in which only a subset of nodes are initially colored, causing other nodes to change colors according to some rules. A set of initially colored nodes that eventually color the entire graph is called a zero forcing set of the graph. Interestingly, a leader set for the strong structural controllability of a network is directly related to the idea of a zero forcing set in the underlying network graph~\cite{monshizadeh2014zero,mousavi2017structural,trefois2015zero}. A set of leaders makes a given network strong structurally controllable if the corresponding nodes in the underlying graph constitute a zero forcing set.

In this paper, we consider the leader selection for \emph{resilient strong structural controllability}, that is, the network remains SSC despite some misbehaving nodes or edges deviating from their normal functionalities. A single misbehaving agent or an edge (due to a fault or an adversary) can severely deteriorate the network's controllability. Considering the susceptibility of networks to adversarial attacks and failures, the resilient SSC is a crucial paradigm in network control systems. We consider three different abnormal behaviors by nodes (agents) and edges. These misbehaving nodes and edges primarily disrupt the zero forcing process in the underlying network graph, affecting the network's SSC. In particular, we examine leak nodes (agents that do not follow the rules of the zero forcing process), non-forcing edges (links that again disrupt the zero forcing process), and removable edges (links that are lost and affect interconnections between nodes). We then study the leader set guaranteeing the network SSC despite a given number of such misbehaving nodes/edges. For this, we utilize the idea of \emph{leaky forcing}, a variant of zero forcing in graphs recently introduced in~\cite{dillman2019leaky,alameda2022leaky}. Our main result shows that the resilient leader selection against one abnormal behavior implies resilience against the other two abnormal behaviors. The main contributions are below:

\begin{itemize}
    \item We study the resilient SSC problem in networks by connecting it to the notion of leaky forcing in graphs. 
    \item We consider three models of misbehaving nodes/edges in networks, including node leaks, non-forcing edges, and removable edges, and formulate leader selection in networks to guarantee that the network remains SSC despite $k$ such abnormal nodes/edges (Section~\ref{sec:RSSC}).
    \item We show that a given leader set is resilient to $k$ misbehaving nodes/edges under one model if and only if it is resilient to the same number of misbehaving nodes/edges in the other model (Section~\ref{sec:equivalence}).
    \item Finally, we discuss computational aspects of selecting leaders for the resilient SSC and also present a numerical evaluation (Section~\ref{sec:computation}).
\end{itemize}


There have been works discussing the impact of nodes/edges addition and deletion on the network's controllability, for instance,~\cite{pirani2022survey,mousavi2017robust,zhang2019minimal,abbas2021edge}. Similarly, some researchers have considered the effects of actuators/input node failures on the overall controllability of systems. Some also propose methods to have a sufficient number of input nodes maintaining network controllability despite such failures ~\cite{pequito2017robust,pu2012robustness,ramos2021robust,zhang2018controllability}. Our work differs from theirs as we consider node/edge misbehavior models beyond deletion. These misbehaving nodes/edges remain within the network disrupting the underlying graph-theoretic process required to ensure SSC. Our goal is to have leader sets guaranteeing SSC despite the existence of any $k$ such nodes/edges within the network. For this, we utilize the idea of leaky forcing in graphs.

The rest of the paper is organized as follows: Section~\ref{sec:prelim} presents preliminary ideas related to network controllability and zero forcing in graphs. Section~\ref{sec:RSSC} explains various threat models and formulates the leader selection problems for the resilient SSC in networks. Section~\ref{sec:equivalence} presents the main result regarding the equivalence of leader selection for resilient SSC under various threat models. Section~\ref{sec:computation} discusses the computational aspects and provide numerical illustrations. Finally, Section~\ref{sec:conclusion} concludes the paper.

\section{Preliminaries}
\label{sec:prelim}
We consider a network of agents modeled by an undirected graph $G = (V,E)$. The node set $V$, and the edge set $E$, represent agents and interconnections between agents, respectively. The edge between nodes $i$ and $j$ is denoted by an unordered pair $(i,j)$. The \emph{neighborhood} of $u$ is $\mathcal{N}_i = \{j\in V|\;(i,j)\in E\}$. 
The \emph{distance} between nodes $i$ and $j$, denoted by $d(i,j)$, is the number of edges in the shortest path between them. For a given graph $G=(V,E)$ with $|V|=n$ nodes, we define a family of symmetric matrices $\mathcal{M}(G)$ as below:
\begin{equation}
\label{eq:graph_matrices}
\begin{split}
\mathcal{M}(G) = \{M\in\mathbb{R}^{n\times n}\;&| \;M = M^\top, \text{ and for }i\ne j,\\
& M_{ij} \ne 0 \Leftrightarrow (i,j) \in E\}.
\end{split}
\end{equation}
Note that the location of zero and non-zero entries in every $M\in\mathcal{M}(G)$ remains the same, and is entirely described by the edge set of $G$.
\subsection{System Model and Graph Controllability}
We consider the following leader-follower system defined over a graph $G=(V,E)$.
\begin{equation}
\label{eq:system}
\dot{x}(t) = Mx(t) + B u(t),
\end{equation}
where $x(t)\in\mathbb{R}^{n}$ is the system state, $u(t)\in\mathbb{R}^m$ is the external input, $M\in\mathcal{M}(G)$ (as in \eqref{eq:graph_matrices}) is the system matrix, and $B\in\mathbb{R}^{n\times m}$ is the input matrix describing the \emph{leader} (input) nodes. If $V = \{v_1,v_2,\cdots,v_n\}$, and  $V' = \{\ell_1,\ell_2,\cdots,\ell_m\}\subseteq V$ is a set of leader nodes, then we define $B$ as follows:
\begin{equation}
\label{eq:B}
[B]_{ij} = \left\lbrace\begin{array}{lll}
1 & \text{if } v_i = \ell_j,\\
0 & \text{otherwise.}
\end{array}\right.
\end{equation}
We observe that the family of matrices $\mathcal{M}(G)$ describes a number of system matrices defined over $G$, including the Laplacian and adjacency matrices. 

For a given graph $G$, system matrix $M\in\mathcal{M}(G)$, and input matrix $B$, the system in \eqref{eq:system} is called \emph{controllable} if there exists an input to drive the system from an arbitrary initial state $x(t_0)$ to an arbitrary final state $x(t_f)$. In this case, we say that $(M,B)$ is a \emph{controllable pair}. A pair $(M,B)$ is controllable if and only if the rank of the \emph{controllability matrix} $\Gamma(M,B)$, defined below, is $|V|=n$ (i.e., full rank).
\begin{equation}
\label{eq:c_matrix}
\Gamma(M,B) = \left[\begin{array}{lllll}
    B & MB & M^2B & \cdots & M^{n-1}B  \\
\end{array}\right].
\end{equation}
Since leader nodes $V'$ define the input matrix $B$, we sometimes abuse the notation slightly and use \emph{$(M,V')$ is controllable} to denote that $(M,B)$ is a controllable pair. 

\begin{definition} \emph{(Network Strong Structural Controllability)}
\label{def:SSC}
A network $G = (V,E)$ with a leader set $V'\subseteq V$ is strong structurally controllable if and only if $(M,V')$ is a controllable pair \emph{for all} $M\in\mathcal{M}(G)$.
\end{definition}

Network strong structural controllability is a stronger notion compared to the (weak) structural controllability, which requires the existence of at least one matrix $M\in\mathcal{M}(G)$ for which $(M,V')$ is a controllable pair.
\subsection{Network Controllability and Zero Forcing in Graphs}
A central problem in network controllability is to compute a minimum leader set $V'\subseteq V$ that makes the network strong structurally controllable (as defined above). The problem is often referred to as the \emph{minimum leader selection} for network controllability. A graph-theoretic characterization of the minimum leader set rendering the network strong structurally controllable is remarkably useful here. Monshizadeh et al. characterized the minimum leader set for network strong structural controllability in \cite{monshizadeh2014zero} using the notion of \emph{zero forcing} in graphs, which is related to the dynamic coloring of nodes. We introduce zero forcing ideas below and then state the main result from \cite{monshizadeh2014zero}.

\begin{definition} (\emph{Zero forcing}) Given a graph $G=(V,E)$ whose nodes are initially colored either \emph{black} or \emph{white}. Consider the following node color changing rule: \emph{If $v\in V$ is colored black and has exactly one white neighbor $u$, change the color of $u$ to black.} Zero forcing is the application of the above rule until no further color changes are possible.
\end{definition}

\begin{definition} \emph{(Force)}
A force is an application of the color changing rule due to which the color of a white node $u$ is changed to black by some black node $v$. We say that $v$ \emph{forced} $u$ and denote it by $v\rightarrow u$.
\end{definition}

\begin{definition} (\emph{Input and derived sets})
For a graph $G = (V,E)$ with an initial set of black nodes $V'\subseteq V$, the \emph{derived set} of $V'$, denoted by $\mathcal{D}(V')$, is the set of black nodes obtained after applying the zero forcing rule exhaustively. The initial set of black nodes $V'$ is the set of \emph{input nodes}. 
\end{definition}

We note that for a given input set $V'$, the derived set $\mathcal{D}(V')$ is unique~\cite{work2008zero}.

\begin{definition} (\emph{Zero forcing set (ZFS) and zero forcing number}) For a graph $G = (V,E)$, an input set $V'\subseteq V$ is a ZFS if $\mathcal{D}(V') = V$ (i.e., all nodes are colored black after the exhaustive application of the zero forcing rule). We denote a ZFS by $Z_0$. The number of nodes in the minimum ZFS is the \emph{zero forcing number} of the graph and denoted by $z_0(G)$.
\end{definition}

Figure~\ref{fig:chronological} illustrates the zero forcing terms defined above.
\begin{figure}[h]
    \centering
    \includegraphics[scale=0.8]{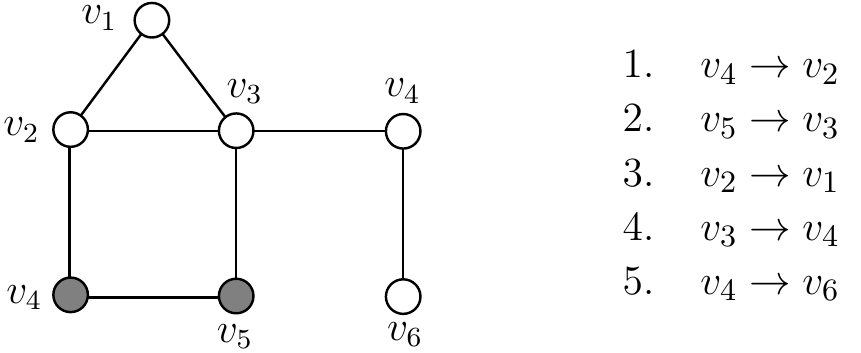}
    \caption{$V'=\{v_4,v_5\}$ is a ZFS of the graph along with a sequence of forces coloring all nodes black.}
    \label{fig:chronological}
\end{figure}

A leader set for the strong structural controllability is closely related to the notion of ZFS of the network graph. A direct consequence of Theorems~IV.4,~IV.8, and Proposition IV.9 in \cite{monshizadeh2014zero} is the following result:

\begin{theorem} \cite{monshizadeh2014zero}
The undirected network $G=(V,E)$ is strong structurally controllable with a leader set $V'\subseteq V$  (as in Definition~\ref{def:SSC}) if and only if $V'$ is a ZFS of $G$, (i.e., $\mathcal{D}(V') = V$).
\end{theorem}

Thus, ZFS in graphs is an important idea from the network controllability perspective and it completely characterizes the leader set for the strong structural controllability of the network. The above results implies that the minimum number of leaders needed for the network strong structural controllability is same as the zero forcing number of the network graph. We note that computing a minimum ZFS and zero forcing number are NP-hrad in general \cite{aazami2008hardness}. However, there are several heuristics to compute a small-sized ZFS, for instance, see \cite{brimkov2019computational,brimkov2021improved,agra2019computational,AbbasACC2022}.

\section{Resilient Strong Structural Controllability (SSC) in Networks}
\label{sec:RSSC}
The equivalence between ZFS and leader set for network SSC is significant. The zero forcing process can be disrupted by some nodes/edges abnormal behaviors or edge failures, and consequently, the network SSC can be adversely impacted. For instance, consider the network in Figure~\ref{fig:exp_1} with the leader set $V' = \{v_1,v_2\}$. If all nodes are normal, the ZF process initiated by the leader set $V'$ will eventually color all nodes in $V$ (i.e., $\mathcal{D}(V') = V$), and the network will be strong structurally controllable with $V'$. However, if $v_5$ behaves abnormally in the sense that it does not force any other node, then the zero forcing process will be hindered and $\mathcal{D}(V') \ne V$ asserting that the network is not strong structurally controllable. Similarly, if we consider all nodes normal, but the edge $(v_5,v_6)$ is removed, the ZF process will again be disrupted, causing $\mathcal{D}(V')\ne V$, which means the network is not strong structurally controllable with $V'$.

\begin{figure}[htb]
    \centering
    \begin{subfigure}[b]{0.155\textwidth}
	\centering
	\includegraphics[scale=0.475]{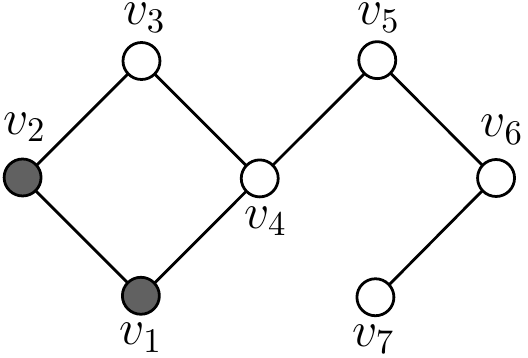}
	\caption{$V' = \{v_1,v_2\}$}
	\end{subfigure}
	\begin{subfigure}[b]{0.16\textwidth}
	\centering
	\includegraphics[scale=0.475]{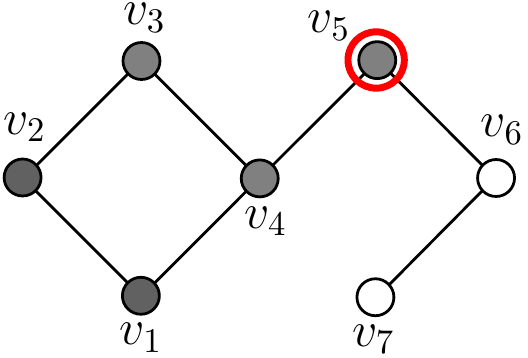}
	\caption{}
	\end{subfigure}
	\begin{subfigure}[b]{0.15\textwidth}
	\centering
	\includegraphics[scale=0.475]{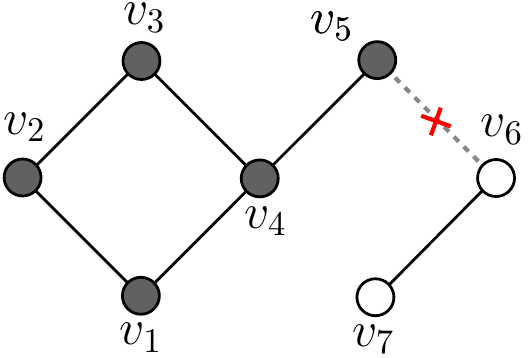}
	\caption{}
	\end{subfigure}
    \caption{(a) $V'$ is a ZFS of $G$. (b) $v_5$ is a misbehaving node not forcing any other node. (c) $(v_5,v_6)$ is a misbehaving edge.}
    \label{fig:exp_1}
\end{figure}
However, if $V' = \{v_1,v_2,v_7\}$, then the network remains SSC despite any single misbehaving node (refusing to force other nodes) or an edge. Thus, the network SSC can be preserved even in the presence of misbehaving nodes or edges through some redundant leader nodes selected carefully. Our goal here is to study: \emph{how we can guarantee resilient network SSC in the face of some misbehaving nodes and edges that might disrupt the zero forcing process and impact the network SSC}.

Next, we will consider three different abnormal node and edge behaviors disrupting the zero forcing process. Then, we present leader selections guaranteeing all nodes in the network get colored due to the zero forcing process despite a certain number of misbehaving nodes and edges, thus achieving the resilient network SSC. Our main result (in Section~\ref{sec:equivalence}) shows that resilience to one type of misbehaving nodes/edges implies resilience to the other kinds of misbehaving nodes/edges. 
\subsection{Failure Models and Resilience Problems}
We consider the following three node and edge misbehaviors that can be caused by the adversarial attack or other abnormality. All of these failures ultimately disrupt the zero forcing process.

\textbf{\emph{1) Leak (non-forcing) nodes:}} A \emph{leak} is a node $v\in V$ that does not force any other node, i.e., considering $v$ to be a leak node that is colored black and has exactly one white neighbor, then $v$ does not force its white neighbor (which it should in case $v$ was normal). A set of all leaks is the \emph{leak set}, denoted by $L\subseteq V$.

The term `leak node' is adapted from \cite{dillman2019leaky}, where such a non-forcing behavior of nodes is introduced. Practically, a leak node can be realized in multiple ways. For instance, if an additional node $\alpha$, which is not a part of the original network and is colored white, becomes adjacent to exactly one node, say $v$, in the network $G$, then $v$ is unable to force any other node in $G$. Figure~\ref{fig:exp_2} illustrates this situation.  

\begin{figure}[htb]
    \centering
    \includegraphics[scale=0.625]{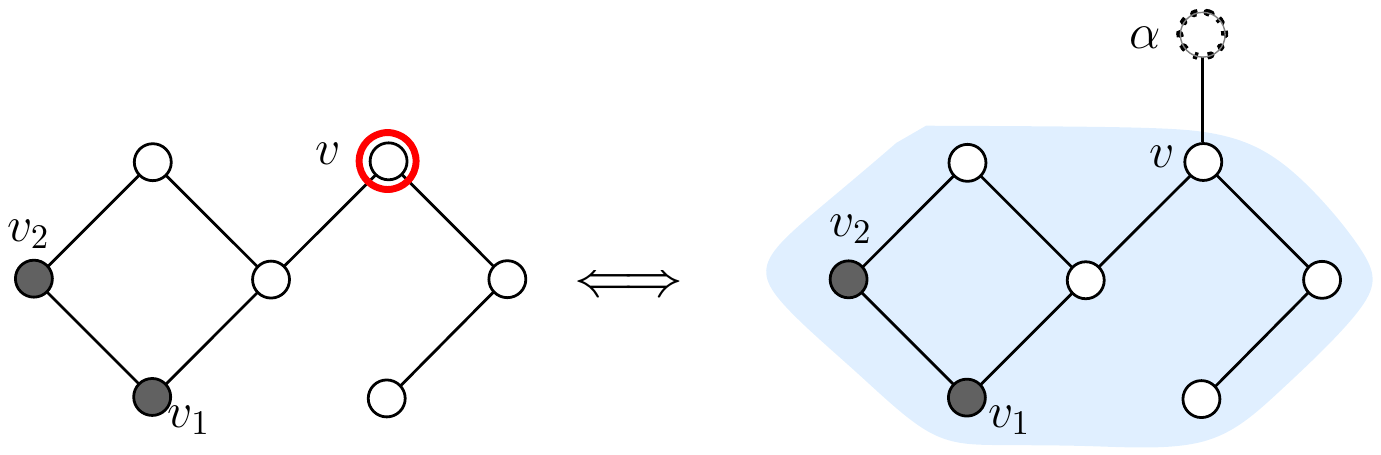}
    \caption{$v$ is a leak node not forcing any other node. Equivalently, an outside node $\alpha$ becomes adjacent to $v$ and prevents $v$ from forcing any node.}
    \label{fig:exp_2}
\end{figure}

Now the \emph{resilience problem} is to have a (minimal) leader set such that all nodes are colored at the end of the zero forcing process despite $\ell$ leak nodes, which are unknown. For a given $\ell$, computing such a leader set is referred to as the \emph{$\ell$-leaky forcing set problem} \cite{dillman2019leaky,alameda2022leaky}. We formally define the leaky derived set and leaky forcing set below:

\begin{definition} (\emph{Leaky derived set}) Given a graph $G=(V,E)$, input set $V'$, and a leak set $L$, then the set of black nodes obtained after applying the zero forcing rule exhaustively while considering the leaks in $L$ is the leaky derived set, denoted by $\mathcal{D}_L(V')$.
\end{definition}

\begin{definition} (\emph{$\ell$-leaky forcing set ($\ell$-LFS)}) An input set $V'\subseteq V$ is an \emph{$\ell$-LFS} if for \emph{any} leak set $L\subset V$ with $\ell$ leaks, $\mathcal{D}_L(V') = V$. In other words, starting with $V'$, all nodes are colored black by iteratively applying the zero forcing rule with \emph{any} $\ell$ leaks. The cardinality of the minimum $\ell$-LFS is the \emph{$\ell$-forcing number} of $G$, denoted by $z_\ell(G)$.
\label{def:lLFS}
\end{definition}

We note that for $\ell = 0$, the $\ell$-forcing number is same as the zero forcing number. Further, we know \cite{dillman2019leaky},
\begin{equation}
    z_0(G) \; \le \; z_1(G) \; \le \; \cdots \; \le \; z_{|V|}(G).
\end{equation}


\textbf{\emph{2) Non-forcing edges:}} Here, we consider edge attacks through which an edge cannot be used by either of its end nodes to force the other end node. We call such an edge as a non-forcing edge. In particular, an edge $(u,v)$ is a \emph{non-forcing edge} if $u$ cannot force $v$, and $v$ cannot force $v$. In this case, the resilience problem is to find the {$\ell$-edge forcing set} defined below:

\begin{definition} (\emph{$\ell$-Edge forcing set ($\ell$-EFS)}) For a given $G=(V,E)$ and a positive integer $\ell$, let $E_\ell\subseteq E$ be an \emph{arbitrary} subset of at most $\ell$ non-forcing edges (i.e., $|E_\ell|\le \ell$). An input set $V'\subseteq V$ is an $\ell$-EFS if there is a zero forcing process that colors all nodes in $V$ without using the edges in $E_\ell$ to force nodes. 
\label{def:lEFS}
\end{definition}

Figure~\ref{fig:1EFS} illustrates the non-forcing edge and $1$-EFS. $V' = \{v_1,v_2\}$ is a ZFS of $G$ in Figure~\ref{fig:1EFS}(a). If the edge between $v_2$ and $v_3$ is non-forcing, then the derived set consists of only three nodes $\{v_1,v_2,v_4\}$ at the end of the zero forcing process. However, if the leader set is $V' = \{v_1,v_2,v_7\}$, then all nodes are colored as a result of the zero forcing process despite any single non-forcing edge.

\begin{figure}[htb]
    \centering
    \begin{subfigure}[b]{0.16\textwidth}
	\centering
	\includegraphics[scale=0.52]{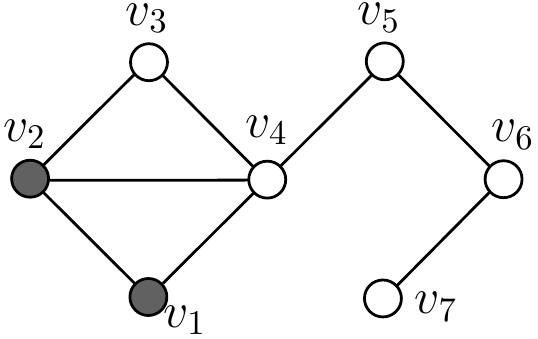}
	\caption{}
	\end{subfigure}
	\begin{subfigure}[b]{0.16\textwidth}
	\centering
	\includegraphics[scale=0.52]{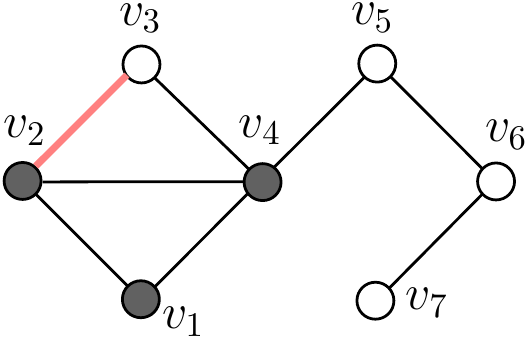}
	\caption{}
	\end{subfigure}
	\begin{subfigure}[b]{0.15\textwidth}
	\centering
	\includegraphics[scale=0.52]{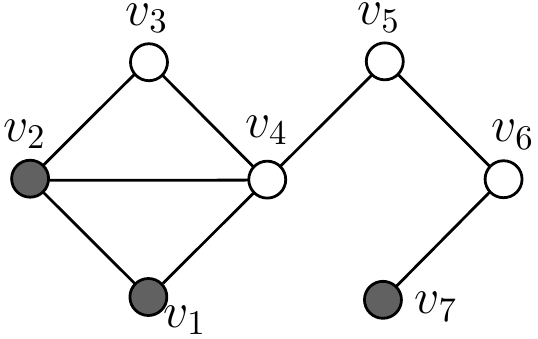}
	\caption{}
	\end{subfigure}
    \caption{(a)~$\{v_1,v_2\}$ is a ZFS given all edges are normal. (b)~$(v_2,v_3)$ is a non-forcing edge. (c)~$\{v_1,v_2,v_7\}$ is a $1$-EFS.}
    \label{fig:1EFS}
\end{figure}

\textbf{\emph{3) Removable edges:}} The third failure model we consider is the one where a maximum of $\ell$ edges are removed from the graph to disrupt the zero forcing process. The corresponding resilience problem is to have enough leaders to guarantee that despite $\ell$ edge removals, the network remains strong structurally controllable, or equivalently all nodes are colored as a result of the zero forcing. In other words, the goal is to find a minimum size \emph{$\ell$-forcing set with removable edges} defined below:

\begin{definition} \emph{$\ell$-forcing set with removable edges ($\ell$-FSR)}
For a given $G=(V,E)$ and $\ell$, consider a subgraph $G' = (V,E')$, where $E'\subseteq E$ and $|E|-|E'|\le \ell$. Then, $V'\subseteq V$ is an $\ell$-FSR of $G$ if $V'$ is a ZFS of \emph{every} such $G'$. Note that an $\ell$-FSR must also be a ZFS of $G$. 
\label{def:lFSR}
\end{definition}

We observe that making an edge non-forcing can be different from removing the edge. For instance, unlike a non-forcing edge, removing an edge can sometimes be useful, as Figure~\ref{fig:FSR} illustrates. If edge $(v_4,v_7)$ in $G$ in Figure~\ref{fig:FSR}(a) is removed, then $V' = \{v_1,v_2\}$ is a ZFS of the resulting graph (Figure~\ref{fig:FSR}(b)), thus, making the network controllable. However, if $(v_4,v_7)$ is a non-forcing edge (Figure~\ref{fig:FSR}(c)), then $\{v_1,v_2\}$ is no longer a ZFS of the network.

\begin{figure}[htb]
    \centering
    \begin{subfigure}[b]{0.16\textwidth}
	\centering
	\includegraphics[scale=0.52]{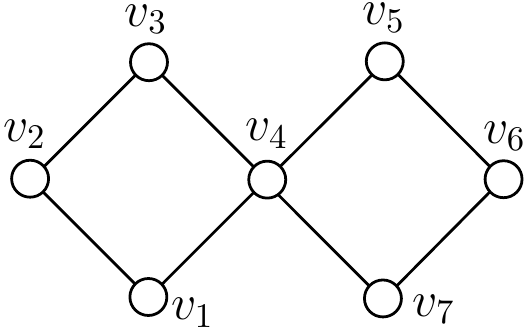}
	\caption{}
	\end{subfigure}
	\begin{subfigure}[b]{0.16\textwidth}
	\centering
	\includegraphics[scale=0.52]{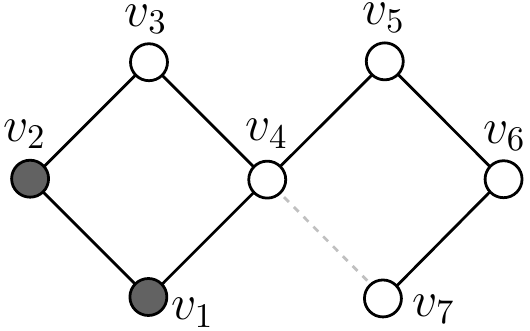}
	\caption{}
	\end{subfigure}
	\begin{subfigure}[b]{0.15\textwidth}
	\centering
	\includegraphics[scale=0.52]{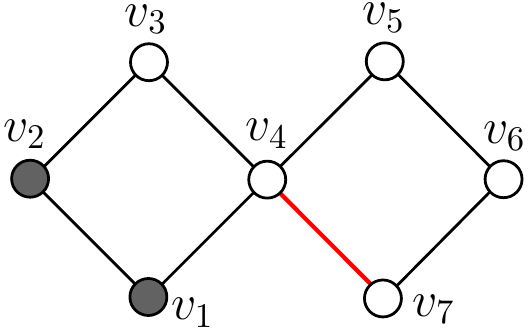}
	\caption{}
	\end{subfigure}
    \caption{(a) A graph $G$. (b) $\{v_1,v_2\}$ becomes a ZFS of $G$ after removing the edge $(v_4,v_7)$. (c) If $(V_4,v_7)$ is a non-forcing edge, then $\{v_1,v_2\}$ is a not a ZFS of $G$.}
    \label{fig:FSR}
\end{figure}

Next, we present the main result showing that a leader set resilient to one misbehavior model is also resilient to the other models.
\section{Equivalence of Resilient Leader Selection for Various Failure Models}
\label{sec:equivalence}
Here, for a given $G = (V,E)$ and $\ell$, we show the equivalence between $\ell$-LFS, $\ell$-EFS and $\ell$-FSR. As a result, we show that a leader set $V'\subseteq V$ ensures resilient controllability against one misbehavior model if and only if it extends resilience to the other two models (discussed above). For instance, a leader set $V'$ that is resilient to $\ell$ non-forcing nodes must also be resilient to $\ell$ non-forcing edges and $\ell$ removable edges simultaneously. We introduce the following terms as in \cite{barioli2010zero,alameda2022leaky}.

\begin{definition}
Consider a graph $G = (V,E)$, input set $V'\subseteq V$ and the corresponding derived set $\mathcal{D}(V')$, then we define the following terms:
\begin{itemize}
    \item A \emph{chronological list of forces} is a list of forces recorded in the order in which they are performed to construct the derived set.
    \item A \emph{forcing process} $F$ is a set of forces containing a chronological ordering of forces through which all nodes in $V$ are colored black (i.e., $\mathcal{D}(V') = V$).
    \item A \emph{forcing chain} is a sequence of forces $v_i\rightarrow v_{i+1}$, for $i = 1,2,\cdots, k-1$. We denote such a forcing sequence by $v_1\rightarrow v_2\rightarrow  \cdots \rightarrow v_{k-1}\rightarrow v_k$.
    \item A \emph{maximal forcing chain}  $v_1\rightarrow v_2\rightarrow  \cdots \rightarrow v_{k-1}\rightarrow v_k$ is a forcing chain such that $v_1\in V'$ and $v_k$ does not force ant other node in $G$.
    \item A \emph{total forcing set} of $V'$, denoted by $\mathcal{F}(V')$, is a set of all possible forces given an input $V'$, i.e., $v_i\rightarrow v_j \in \mathcal{F}(V')$ if there is a forcing process in $G$ containing $v_i\rightarrow v_j$.
    \item A \emph{total forcing set with leaks} $L$ and input set $V'$, denoted by $\mathcal{F}_L(V')$, is a set of all possible forces given an input set $V'$ and leaks $L$. In other words, if $v_i\rightarrow v_j \in \mathcal{F}_L(V')$, then $v_i\notin L$ and there is a forcing process containing the force $v_i\rightarrow v_j$.
\end{itemize}
\end{definition}

Next, consider $G=(V,E)$, a ZFS $V'\subset V$, and a forcing process $F$ with $V'$. Then, we define the following notations: 
\begin{equation}
\label{eq:com_1}
F(V') = \{x\rightarrow y \in F: \; y\notin V'\}.
\end{equation}

Similarly, 
\begin{equation}
\label{eq:com_1}
F\backslash F(V') = \{x\rightarrow y \in F: \; y\in V'\}.
\end{equation}

We now state some results from~\cite{alameda2022leaky} that will be used later.
\begin{lem}\cite{alameda2022leaky}
\label{lem:combined_ZFP}
Consider $G=(V,E)$ and a ZFS $V'$. Let $F$ and $F'$ be two forcing processes. Then, $(F\backslash F(\tilde{V}))\cup F'(\tilde{V})$ is a forcing process with $V'$ for any $\tilde{V}$ obtained from $V'$ using $F$. Here, 
\end{lem}

\begin{lem} \cite{alameda2022leaky}
\label{thm:1-LFS_previous}
A set $V'$ is a $1$-LFS if and only if $\forall v\in V\setminus V'$, there exists $x\rightarrow v \in \mathcal{F}(V')$, $y\rightarrow v\in\mathcal{F}(V')$, where $y\notin x$.
\end{lem}

\begin{theorem}\cite{alameda2022leaky}
\label{thm:ell-LFS_previous}
A set $V'$ is a $\ell$-LFS if and only if $V'$ is an $(\ell-1)$-LFS for every set of of $\ell-1$ leaks $L$ and $v\in V\setminus V'$, there exists $x\rightarrow v \in \mathcal{F}_L(V')$, $y\rightarrow v\in\mathcal{F}_L(V')$, where $y\notin x$.
\end{theorem}

\subsection{Main Result}
Our main result here is to show the following:

\begin{theorem}
\label{thm:main_result}
Given a graph $G=(V,E)$, input set $V'\subseteq V$, and a positive integer $\ell\le |V|$, the following statements are equivalent:
\begin{enumerate}
    \item $V'$ is an $\ell$-LFS.
    \item $V'$ is an $\ell$-EFS.
    \item $V'$ is an $\ell$-FSR.
\end{enumerate}
\end{theorem}

We recall that notions of $\ell$-LFS, $\ell$-EFS, and $\ell$-FSR are explained in Definitions~\ref{def:lLFS}, \ref{def:lEFS}, and \ref{def:lFSR}, respectively. To prove Theorem~\ref{thm:main_result}, we need some intermediate results.

\begin{lem}
\label{lem:1EFS}
If $V'$ is a $1$-LFS of $G = (V,E)$, then for every edge $e = (u,v)\in E$, there exists a zero forcing process $F_e$ that does not use $e$, i.e., $u\rightarrow v \notin F_e$ and  $v\rightarrow u \notin F_e$.
\end{lem}
\begin{proof}
Consider a forcing process $F$ which contains $u\rightarrow v$. By Lemma~\ref{thm:1-LFS_previous}, there exists a forcing process $F'$ such that $x\rightarrow v$, where $x\ne u$. It suffices to show that there exists a forcing process $F_e$ such that $x\rightarrow v$ and $y \rightarrow u$, where $y\ne v$. If $u\in V'$, then $F_e = F'$. So, assume $u\notin V'$. Now, again consider $F$ and let $\tilde{V}$ be the set of black vertices obtained from $V'$ such that $u\rightarrow v$ is valid given $\tilde{V}$, but $v\notin \tilde{V}$. Note that at this point, $u$ and all of its neighbors except $v$ are colored black through $F$. Moreover, a node can only force one node. Thus, $u$ will not force any node in $F'(\tilde{V})$, i.e., $u\rightarrow q \notin F'(\tilde{V}),\; \forall q \in V$. Now, consider $F_e = (F\backslash F(\tilde{V})) \cup F'(\tilde{V})$. Note that $u\rightarrow v \notin F_e$ and $v\rightarrow u \notin F_e$. By Lemma~\ref{lem:combined_ZFP}, $F_e$ is a forcing process with input $V'$, which proves the desired claim. 
\end{proof}
\begin{lem}
\label{lem:EFStoLFS}
If $V'$ is a $1$-EFS, then $V'$ is a $1$-LFS
\end{lem}
\begin{proof}
Let $F$ be a forcing process with $V'$ as input nodes. Consider $u\in V$ to be be an arbitrary fixed leak. If $u$ is an end node of a forcing chain, then $u$ does not force any node. So, we assume that $u\rightarrow v \in F$ for some $v$. Since $V'$ is $1$-EFS, there exists another forcing process $F'$ such that $v$ is not forced by $u$. Let $x\rightarrow v$, for some $x\ne u$. Now, consider $F$ till the point $u\rightarrow v$ becomes valid but $v$ is not colored black. Let $\tilde{V}$ be the set of black nodes till this point. Note that $v\notin \tilde{V}$. Note that all nodes in $\mathcal{N}[u]\setminus \{v\}$ are colored black and $v$ is the only node that $u$ can force. Consider $F_u = (F\backslash F(\tilde{V}))\cup F'(\tilde{V})$. By Lemma~\ref{lem:combined_ZFP}, $F_u$ is a valid forcing process with input nodes $V'$ coloring all nodes black. Note that $u$ is not forcing $v$ in $F_u$, and hence $u$ is not forcing any node in $F_u$. Thus, $V'$ is a $1$-LFS, which is the desired claim.
\end{proof}
\begin{lem}
\label{lem:intermediate}
Let $V'$ be an $(\ell-1)$-EFS, $E_\ell$ be a set of $\ell$ non-forcing edges, and $\mathcal{D}(V')$ be the derived set after forcing. If $e = (u,v) \in E_\ell$, then $u$ and $v$ can not be white simultaneously. Moreover, if exactly one end node of $e$, say $u$, is black, then $N[u]\setminus\{v\}\subseteq\mathcal{D}(V')$.
\end{lem}

\begin{proof}
If both end nodes of $e\in E_\ell$ are white, then none of the end nodes can force the other end node. Thus, zero forcing behavior of $V'$ does not change even if $e$ is not a non-forcing edge. So, if we consider $E_\ell\setminus e$ as the set of non-forcing edges, $\mathcal{D}(V')\ne V$, implying that $V'$ is not an $(\ell-1)$-EFS, which is a contradiction. Similarly, let $e=(u,v)\in E_\ell$ be such that $u\in\mathcal{D}(V')$ and $v\notin\mathcal{D}(V')$. Assume that $x\in N(u)$ is white and $x\ne v$. Since $u$ has two white neighbors, $u$ can not force any node (including $v$) even if $e$ is not a non-forcing edge. Again, considering $E_\ell\setminus e$ as the set of non-forcing edges will give $\mathcal{D}(V')\ne V$. It means $V'$ is not an $(\ell-1)$-EFS, which is a contradiction.
\end{proof}


\begin{lem}
\label{lem:inter_2}
Let $V'$ be an $(\ell-1)$-LFS and $L$ be a set of $\ell$ leaks. Then $L\subseteq \mathcal{D}(V')$, where $\mathcal{D}(V')$ is a derived set with leaks. Also, each $v\in L$ has at most one white neighbor. 
\end{lem}
\begin{proof}
Assume $v\in L$ is white, i.e., $v\notin\mathcal{D}(V')$. A white leak node does not change the zero forcing behavior of the black nodes. So, we consider $L' = L\setminus\{v\}$. Since $V'$ is an $(\ell-1)$-LFS, so all nodes should be black for any $(\ell-1)$ leaks. However, all nodes are not black, which is a contradiction. For the second part, assume that there is leak node $v\in L$ that is colored black and has two white neighbors. A black node with two white neighbors can not force any node, so we consider $L' = L\setminus\{v\}$ as a set of $(\ell-1)$ leaks. Since $V'$ is $(\ell-1)$-LFS, it should color all nodes black with leaks in $L'$, which is not the case. Hence a contradiction arises, proving the desired claim. 
\end{proof}

Next, we show the equivalence between $\ell$-LFS and $\ell$-EFS.

\begin{theorem} 
\label{thm:LSFEFS}
For a given $\ell$ and $G = (V,E)$, $V'\subseteq V$ is an $\ell$-LFS if and only if $V'$ is an $\ell$-EFS.
\end{theorem}
\begin{proof} 
($\ell$-LFS $\rightarrow \ell$-EFS) We will prove using induction on $\ell$. From Lemma~\ref{lem:1EFS}, if $V'$ is $1$-LFS, then it is $1$-EFS. So, our induction hypothesis is, if $V'$ is $(\ell-1)$-LFS, then $V'$ is $(\ell-1)$-EFS. Now, assume that $V'$ is $\ell$-LFS. Thus, $V'$ must be $(\ell-1)$-LFS implying that it is also $(\ell-1)$-EFS (by our induction hypothesis). Let $E_\ell$ be a set of $\ell$ non-forcing edges. By Lemma~\ref{lem:intermediate}, there is a forcing process $F$ such that for each $e = (u,v)\in E_\ell$, either both end nodes $u$ and $v$ are black, or one end node, say $u$, is black with $N[u]\setminus\{v\}$ also colored black. Let $\tilde{V}$ be the set of black nodes with the forcing process $F$. Next, for each $e\in E_\ell$, we consider one of its black end node as a leak, and denote the set of leaks by $L$. There will be at most $\ell$ leaks. We observe that a black colored leak node can be ignored and deleted without altering the zero forcing behavior of the other black nodes. So, we consider $G' = G\setminus \{L\}$. Since $G$ is $\ell$-LFS with $V'$, $\tilde{V}\setminus L$ is a ZFS of $G'$. As a result, all nodes in $V$ are colored black while considering $E_\ell$ as non-forcing edges, implying $V'$ is an $\ell$-EFS.


($\ell$-EFS $\rightarrow \ell$-LFS) Again, we will use induction on $\ell$. For $\ell = 1$, if $V'$ is $1$-EFS, then it is $1$-LFS (by Lemma~\ref{lem:EFStoLFS}). Our induction hypothesis is that $V'$ is $(\ell-1)$-EFS implies it to be $(\ell-1)$-LFS. Now assume $V'$ to be $\ell$-EFS. It means $V'$ is $(\ell-1)$-EFS and hence, $(\ell-1)$-LFS (by the induction hypothesis). Consider $L$ to be a set of $\ell$ leaks. By Lemma~\ref{lem:inter_2}, these leaks are colored black and each of them has at most one white neighbor. Next, we consider the edge between the leak and its white neighbor as a non-forcing edge. There will be at most $\ell$ such non-forcing edges, which we denote by $E_\ell$. Since one end node (leak) of each non-forcing edge is colored black and a leak has at most one white neighbor (which is the other end node of the non-forcing edge), we can safely delete the non-forcing edge without affecting the zero forcing behavior of the black node. Thus, we get $G' = G\setminus E_\ell$. Since $V'$ is $\ell$-EFS, it means $V'$ is a ZFS of $G'$. Thus, all nodes will be colored black despite $\ell$ leaks. Hence, $V'$ is $\ell$-LFS, which proves the desired claim.
\end{proof}

\begin{lem}
\label{lem:RMSEFS}
If $V'$ is $1$-FSR then it is $1$-EFS.
\end{lem}

\begin{proof}
 By contraposition, let $V'$ be a ZFS that is not a $1$-EFS. It means there is an edge, say $(u,v)$, such that every forcing process must use the edge, i.e., $(u,v)$ must be a forcing edge in any forcing process. Without the loss of generality, we assume $u$ forces $v$.
It implies that all the black neighbors of $v$ have at least two white neighbors and $u$ has only one white neighbor $v$. Thus, by removing edge $(u,v)$, $v$ can not be forced. Hence, $V'$ is not a $1$-FSR, which is the desired claim. 
\end{proof}

In the following, we show the equivalence of $\ell$-EFS and $\ell$-FSR.

\begin{theorem}
\label{thm:EFSRMS}
$V'$ is $\ell$-EFS if and only if $V'$ is $\ell$-FSR.
\end{theorem}

\begin{proof}
($\ell$-EFS $\rightarrow$ $\ell$-FSR) We first note that if $F$ is a ZFP with $V'$ as a ZFS, then there are at most $n-1$ edges used in $F$. If we remove edges not used in $F$ to get $G'$, then $V'$ will still be a ZFS of $G'$. Thus, if $V'$ is an $\ell$-EFS, it means for any edge set $E_\ell\subseteq E$, where $|E_\ell|\le \ell$, there exist a forcing process with $V'$, say $F_{\ell}$, coloring all nodes black without using edges in $E_\ell$. Since edges in $E_\ell$ are not used in $F_{\ell}$, we can remove them from $G$ while maintaining $V'$ to be a ZFS of $G' = (V,E\setminus E_{\ell})$, i.e., $F_{\ell}$ is a forcing process of $G'$ with $V'$, implying $V'$ is an $\ell$-FSR of $G$.

($\ell$-FSR $\rightarrow$ $\ell$-EFS) We will prove using induction on $\ell$. For $\ell=1$, if $V'$ is $1$-FSR, then it is $1$-EFS by Lemma~\ref{lem:RMSEFS}. Thus, we make the induction hypothesis, if $V'$ is $(\ell-1)$-FSR, then it is $(\ell-1)$-EFS. Assuming $V'$ to be $\ell$-FSR, we need to show that $V'$ is $\ell$-EFS. Let $E_\ell$ be a set of $\ell$ non-forcing edges. Since $V'$ is $(\ell-1)$-EFS (by the induction hypothesis), we apply forces such that each edge $e = (u,v) \in E_\ell$ satisfies one of the two conditions (by Lemma~\ref{lem:intermediate}):
(i) Both end nodes of $e$ are colored black. 
(ii) If one node, say $u$, is black, then all nodes in $N[u]\setminus\{v\}$ are also colored black. Note that the removal of edges in both cases (i) and (ii) will not change the zero forcing behavior of the black nodes. Thus, we remove these $\ell$ edges. Since $V'$ is $\ell$-FSR, it means that all nodes will be colored black at the end of the forcing process. Thus, all nodes are colored black in spite of $\ell$ non-forcing edges, i.e., $V'$ is $\ell$-EFS, which is the desired result.
\end{proof}

A direct corollary of Theorems~\ref{thm:LSFEFS} and \ref{thm:EFSRMS} is Theorem~\ref{thm:main_result} entailing that resilience to one type of misbehaving nodes/edges implies resilience to other kinds of misbehaving nodes/edges.




\section{Computation and Numerical Illustration}
\label{sec:computation}
Computing a minimum $\ell$-LFS, and therefore, $\ell$-EFS and $\ell$-FSR, are NP hard problems (since minimum ZFS is NP-hard). Here, first, we illustrate the characterization of $1$-LFS provided in Theorem~\ref{thm:1-LFS_previous}. Then, we present a greedy algorithm to compute a small-sized $1$-LFS and numerically evaluate it. A greedy algorithm for any $\ell$ can be designed using a similar approach and utilizing the characterization in Theorem~\ref{thm:ell-LFS_previous}. 

For a given input set $V'$, we define $\mathcal{Q}(V')$ to be the set of non-input nodes, each of which can be forced by at least two distinct nodes. More precisely,
\begin{equation}
\begin{split}
\label{eq:Q}
\mathcal{Q}(V') = \{v\in V\setminus V' : \; \; & \exists \; x\rightarrow v \in \mathcal{F}(V') \text{ and}\\ &y\rightarrow v \in \mathcal{F}(V'), \text{ and } x\ne y\}.
\end{split}
\end{equation}

By Theorem~\ref{thm:1-LFS_previous}, $V'$ is $1$-LFS if and only if $\mathcal{Q}(V') = V\setminus V'$. Consider the graph in Figure~\ref{fig:R1_example}, where $V' = \{v_1,v_2,v_4,v_6\}$ is a $1$-LFS. To verify this, we need to show that for each $v\in V\setminus V'$, there exist zero forcing processes such that $v$ is forced by at least two distinct nodes. Table~\ref{table:only} outlines multiple zero forcing processes. We observe that for each $v\in V\setminus V'$, there exist zero forcing processes where $v$ is forced by two distinct forcers. For instance, $v_3$ is forced by $v_2$ and $v_8$ in ZFP 1 and ZFP 2, respectively. Similarly, $v_5$ is forced by $v_9$ and $v_3$ in ZFP 1 and ZFP 3, respectively. 

\begin{figure}[htb]
    \centering
    \includegraphics[scale=0.65]{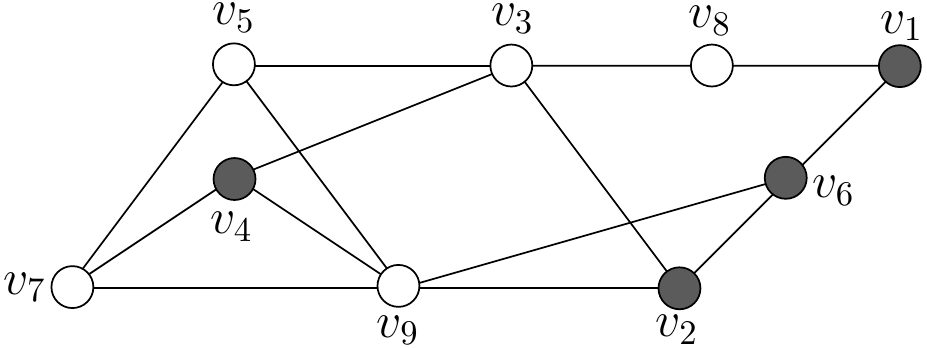}
    \caption{$V' = \{v_1,v_2,v_4,v_6\}$ is a $1$-LFS of $G$.}
    \label{fig:R1_example}
\end{figure}

\begin{table*}[h]
\begin{center}
\begin{tabular}{ |c|c|c|c|c|c| } 
 \hline
 ZFP 1 & ZFP 2 & ZFP 3 & ZFP 4 & ZFP 5 & ZFP 6\\ 
 \hline
  $\begin{array}{c}v_1\rightarrow v_8 \\ v_2\rightarrow v_3\\
 v_4\rightarrow v_7\\
 v_6\rightarrow v_9 \rightarrow v_5\end{array}$
 & 
  $\begin{array}{c}v_1\rightarrow v_8\rightarrow v_3 \\ 
  v_2\\
 v_4\rightarrow v_7\\
 v_6\rightarrow v_9 \rightarrow v_5\end{array}$
 & 
 $\begin{array}{c}v_1\rightarrow v_8 \\ 
  v_2\rightarrow v_3\rightarrow v_5\\
 v_4\rightarrow v_7\\
 v_6\rightarrow v_9 \end{array}$
&
 $\begin{array}{c}v_1\rightarrow v_8 \\ 
  v_2\rightarrow v_3\rightarrow v_5\\
 v_4\\
 v_6\rightarrow v_9 \rightarrow v_7\end{array}$
 & 
 $\begin{array}{c}v_1\\
 v_2\rightarrow v_3 \rightarrow v_8\\ 
  v_4\rightarrow v_7 \\
 v_6\rightarrow v_9 \rightarrow v_5\end{array}$
 &
 $\begin{array}{c}v_1\rightarrow v_8\rightarrow v_3 \rightarrow v_5\\ 
  v_2\rightarrow v_9 \\
 v_4\rightarrow v_7\\
 v_6\end{array}$
 \\
 \hline
\end{tabular}
\end{center}
\caption{Multiple zero forcing processes (ZFP) with a given $V'=\{v_1,v_2,v_4,v_6\}$ for the graph in Figure~\ref{fig:R1_example}. For each $v\in V\setminus V'$, there are at least two ZFP such that $v$ is forced by distinct nodes in such ZFP.}
\label{table:only}
\end{table*}

Algorithm~1 presents a \emph{greedy} heuristic to compute a $1$-LFS. The main idea is to iteratively include nodes in the leader set $V'$ to maximize the size of $\mathcal{Q}(V')$ (as in \eqref{eq:Q}) until $\mathcal{Q}(V') = V\setminus V'$. Since every $1$-LFS is a ZFS, we initialize $V'$ in Algorithm~1 with a ZFS. As a result of this greedy selection, $V'$ might contain some redundant nodes. Lines~$9-14$ in Algorithm~1 remove such redundant nodes to reduce the size of $1$-LFS.

{\small
\begin{center}
\begin{tabular}{l l}
\rule[0.08cm]{8.5cm}{0.03cm}\\
\textbf{Algorithm~1: Greedy Heuristics for 1-LFS}\\
\rule[0.08cm]{8.5cm}{0.02cm}\\
\mbox{\small $\;1:\;$} \textbf{given:} $G = (V,E)$, $|V| = n$\\
\mbox{\small $\;2:\;$} \textbf{initialization:}
$V'=\emptyset$\\
\mbox{\small $\;3:\;$} Compute ZFS $Z_{\textbf{0}}$, and assign $V' = Z_{\textbf{0}}$\\
\mbox{\small $\;4:\;$} Compute $\mathcal{Q}(V')$\\
\mbox{\small $\;5:\;$} \textbf{while} {$|\mathcal{Q}(V')| < n-|Z_{\textbf{1}}|$} \\
\mbox{\small $\;6:$}\hspace{0.25in}$v^\ast = \argmax\limits_{v\in V\setminus (V'\cup \mathcal{Q}(V'))} {\mathcal{Q}(V'\cup\{v\})}$ \\
\mbox{\small $\;7:$}\hspace{0.25in}$V' = V'\cup\{v^\ast\}$ \\
\mbox{\small $\;8:\;$} \textbf{end while} \\
\mbox{\small $\;\;\;\;\;\;\;$}\texttt{-------- removing redundancies --------} \\
\mbox{\small $\;9:$} \; $Z = V'$\\
\mbox{\small $\;10:\;$} \textbf{for} {all $v\in Z$} \\
\mbox{\small $\;11:$}\hspace{0.25in}\textbf{if} {$|\mathcal{Q}(V'\setminus \{v\})| = n - |V'|$ - $1$} \\
\mbox{\small $\;12:$}\hspace{0.5in}$V' = V'\setminus \{v\}$ \\
\mbox{\small $\;13:$}\hspace{0.25in}\textbf{end if} \\
\mbox{\small $\;14:\;$} \textbf{end for} \\
\mbox{\small $\;15:\;$} \textbf{return} {$V'$} \\
\rule[0.08cm]{8.5cm}{0.02cm}\\
\end{tabular}
\end{center}
}



Figure~\ref{fig:numerical} compares the greedy heuristic with the optimal solution for Erd\"os-R\'enyi (ER) and Barab\'asi-Albert (BA) random graphs with $n = 20$ nodes. In ER graphs, any two nodes are adjacent with a probability $p$. BA graphs are obtained by attaching a new node (one at a time) to an existing graph through $m$ edges using a preferential attachment model. The optimal solution is computed using an exhaustive search. Also, each point on the plots averages $15$ randomly generated instances. Figures~\ref{fig:numerical}(a) and (c) plot the size of $1$-LFS as a function of $m$ in BA graphs, and as a function of $p$ in ER graphs, respectively. Figures~\ref{fig:numerical}(b) and (d) plot the time taken by the greedy and optimal solutions to compute $1$-LFS in BA and ER graphs, respectively. We observe that greedy and optimal solutions are very close, and the difference between the two solutions reduces as the graphs become dense. However, the difference in time to compute greedy and optimal solutions is orders of magnitude, even in small-sized graphs. We can design a similar greedy heuristic to compute an $\ell$-LFS for $\ell>1$; however, the time complexity will increase significantly with increasing $\ell$, and even the greedy heuristic will become inefficient for higher $\ell$. Thus, more efficient heuristics are needed to compute $\ell$-LFS for large $\ell$ values. 

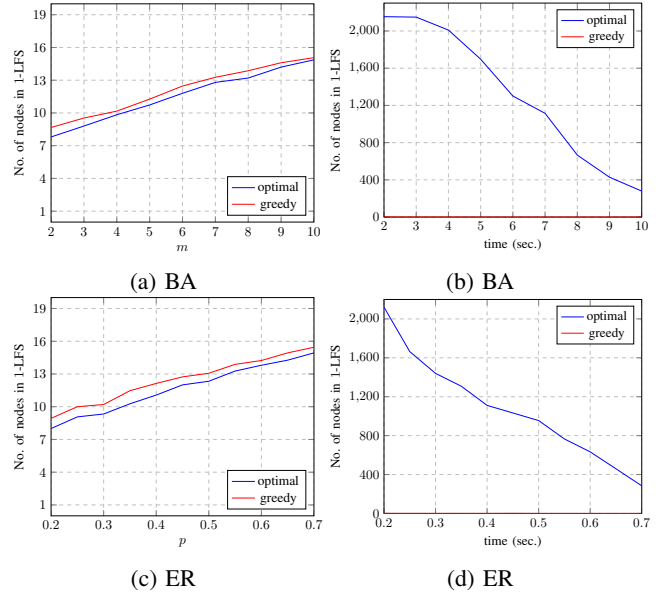
\begin{figure}[h]
\centering
\begin{subfigure}[b]{0.23\textwidth}
	\centering
\begin{tikzpicture}[scale=0.51]
\begin{axis}[
    xlabel={$m$},
    ylabel={No. of nodes in $1$-LFS},
    xmin=2, xmax=10,
    ymin=0, ymax=20,
    xtick={2,3,4,5,6,7,8,9,10},
    ytick={1,4,7,10,13,16,19},
    legend pos=south east,
    xmajorgrids=true,
    ymajorgrids=true,
    grid style=dashed,
]
\addplot[
    color=blue,
    ]
    coordinates {
     (2,7.8)(3,8.8)(4,9.833)(5,10.733)(6,11.8)(7,12.8)(8,13.2)(9,14.2)(10,14.8667)
    };
\addplot[
    color=red,
    ]
    coordinates {
    (2,8.6667)(3,9.5333)(4,10.1667)(5,11.2667)(6,12.4667)(7,13.2667)(8,13.8667)(9,14.6)(10,15.0667)
    };
    \legend{optimal,greedy}  
\end{axis}
\end{tikzpicture}
\caption{BA}
\end{subfigure}
\begin{subfigure}[b]{0.23\textwidth}
	\centering
\begin{tikzpicture}[scale=0.5]
\begin{axis}[
    xlabel={time (sec.)},
    ylabel={No. of nodes in $1$-LFS},
    xmin=2, xmax=10,
    ymin=0, ymax=2300,
    xtick={2,3,4,5,6,7,8,9,10},
    ytick={0,400,800,1200,1600,2000},
    legend pos=north east,
    xmajorgrids=true,
    ymajorgrids=true,
    grid style=dashed,
]
\addplot[
    color=blue,
    ]
    coordinates {
     (2,2154)  (3,2149)   (4,2010)   (5,1697)   (6,1301)   (7,1115) (8,666) (9,428) (10,279)
    };
\addplot[
    color=red,
    ]
    coordinates {
    (2,0.8168)    (3,0.6748)    (4,0.62)    (5,0.3861)   (6,0.2579)    (7,0.2147) (8,0.1911) (9,0.15) (10,0.138)
    };
    \legend{optimal,greedy}  
\end{axis}
\end{tikzpicture}
\caption{BA}
\end{subfigure}
\\
\begin{subfigure}[b]{0.23\textwidth}
	\centering
\begin{tikzpicture}[scale=0.51]
\begin{axis}[
    xlabel={$p$},
    ylabel={No. of nodes in $1$-LFS},
    xmin=0.2, xmax=0.7,
    ymin=0, ymax=20,
    xtick={0.2,0.3,0.4,0.5,0.6,0.7},
    ytick={1,4,7,10,13,16,19},
    legend pos=south east,
    xmajorgrids=true,
    ymajorgrids=true,
    grid style=dashed,
]
\addplot[
    color=blue,
    ]
    coordinates {
     (0.2000,8)(0.2500,9.0667)(0.3000,9.3333) (0.3500,10.2667)(0.400,11.0667)(0.4500,12)(0.500,12.3333)(0.55,13.2667)(0.6,13.8)(0.65,14.2667)(0.7,14.9333)
    };
\addplot[
    color=red,
    ]
    coordinates {
     (0.2000,8.9333)(0.2500,10)(0.3000,10.2) (0.3500,11.4667)(0.400,12.1333)(0.4500,12.7333)(0.500,13.0667)(0.55,13.88)(0.6,14.2333)(0.65,14.9333)(0.7,15.45)
    };
    \legend{optimal,greedy}  
\end{axis}
\end{tikzpicture}
\caption{ER}
\end{subfigure}
\begin{subfigure}[b]{0.23\textwidth}
	\centering
\begin{tikzpicture}[scale=0.5]
\begin{axis}[
    xlabel={time (sec.)},
    ylabel={No. of nodes in $1$-LFS},
    xmin=0.2, xmax=0.7,
    ymin=0, ymax=2200,
    xtick={0.2,0.3,0.4,0.5,0.6,0.7},
    ytick={0,400,800,1200,1600,2000},
    legend pos=north east,
    xmajorgrids=true,
    ymajorgrids=true,
    grid style=dashed,
]
\addplot[
    color=blue,
    ]
     coordinates {
     (0.2000,2121)(0.25,1664)(0.30,1440) (0.35,1308)(0.40,1111)(0.500,955)(0.55,766)(0.6,633)(0.65,459)(0.7,284)
    };
\addplot[
    color=red,
    ]
    coordinates {
    (0.2,0.8319)(0.25,0.6007)(0.3,0.5181)(0.35,0.4489)(0.4,0.3105)(0.45,01.2040)(0.5,0.1887)(0.55,0.1673)(0.6,0.14)(0.65,0.14)(0.7,0.12)
    };
    \legend{optimal,greedy}  
\end{axis}
\end{tikzpicture}
\caption{ER}
\end{subfigure}
\caption{Comparison of optimal and greedy heuristic (Algorithm~1) for the computation of $1$-LFS in Erd\"os-R\'enyi (ER) and Barab\'asi-Albert (BA) random graphs. }
\label{fig:numerical}
\end{figure}

\section{Conclusion}
\label{sec:conclusion}
We studied the problem of maintaining SSC in networks despite misbehaving nodes and edges. We considered different models of misbehaving nodes/edges aiming to disrupt the zero forcing process in graphs, consequently deteriorating the network SSC. 
We showed that resilience to one type of misbehaving nodes/edges is possible if and only if resilience to other threat types is ensured. In other words, for a given leader set, if a network is strong structurally controllable despite $k$ misbehaving nodes/edges under one threat model, then the network remains to be strong structurally controllable in the presence of the same number of misbehaving nodes/edges under the other threat models. We also discussed the leader selection guaranteeing resilient SSC. Resilience is expensive in the context of SSC because many extra leaders are needed to nullify the effect of a small number of misbehaving nodes and edges. 

In the future, we aim to address this problem by safeguarding a subset of nodes and edges, i.e., making them less susceptible to faults and adversarial attacks, and then analyzing how the number of leaders for resilient SSC is impacted in the presence of such `trusted' nodes and edges. Since it would be expensive or infeasible to make every node/edge trusted, we would need to select them in the network carefully. Another interesting direction in this domain requiring more exploration is designing efficient heuristics to compute leaders for resilient SSC.


\bibliographystyle{IEEEtran} 
\bibliography{references}  
\end{document}